%% file: main.tex
\documentclass[letterpaper,11pt]{article}
\usepackage[margin=1in]{geometry}

\usepackage[utf8]{inputenc}

\long\def\commentbegin #1\commentend{}

\usepackage{balance}
\usepackage{algorithm}
\usepackage[noend]{algpseudocode}
\usepackage{ifthen}
\usepackage{comment}
\usepackage{amsfonts}
\usepackage{amsmath}
\usepackage{amssymb}
\usepackage{amsthm}
\usepackage{graphicx}
\usepackage{bbm}
\usepackage{booktabs}

\usepackage{hyperref}
\usepackage{url}

\usepackage{pstricks,pst-node,pst-tree}
\usepackage{soul}
\usepackage{subcaption}
\usepackage{forest}

\usepackage{color}
\usepackage{mathrsfs}
\usepackage[normalem]{ulem}
\usepackage{paralist}

\newtheorem{theorem}{Theorem}[section]
 \newtheorem{lemma}[theorem]{Lemma}

 \newtheorem{corollary}[theorem]{Corollary}

\renewcommand{\epsilon}{\varepsilon}

\renewcommand{\geq}{\geqslant}

\renewcommand{\leq}{\leqslant}

\newcommand{\Prob}[1]{\hbox{\rm I\kern-2pt P}\left[#1\right]}

\def\polylog{\operatorname{polylog}}

\newcommand{\IncreaseExpansion}{\textsc{IncreaseExpansion}}
\newcommand{\CreateExpander}{\textsc{CreateExpander}}
\newcommand{\ExpanderDegreeReduction}{\textsc{ExpanderDegreeReduction}}

\newcommand{\incVect}{\textbf{i}_G}
\newcommand{\sketchVect}{\textbf{s}_G}
\newcommand{\PushSum}{\textsc{Push-Sum}}

\newcommand{\AggregateSketchVector}{\textsc{Aggregate-Sketch-Vector}}

\newcommand{\E}{\mathbbm{E}}
\newcommand{\Ind}{\mathbbm{1}}
\newcommand{\V}{\mathcal{V}}
\def\st{k} 

\def\ShowComment{True}
\ifdefined\ShowComment
\def\billy#1{{\color{green}\underline{\textsf{Billy:}}} {\color{blue} \emph{#1}}}
\def\gopal#1{{\color{red}\underline{\textsf{Gopal:}}} {\color{blue} \emph{#1}}}
\def\fabien#1{{\color{orange}\underline{\textsf{Fabien:}}} {\color{blue} \emph{#1}}}
\else
\def\billy#1{}
\def\gopal#1{}
\def\fabien#1{}
\fi

\DeclareMathAlphabet{\mathsc}{OT1}{cmr}{m}{sc}

\usepackage{ulem}

\title{Time- and Communication-Efficient  Overlay Network Construction via Gossip}
\date{}

\newcommand*\email[1]{\href{mailto:#1}{\url{#1}}}
\author{
Fabien Dufoulon\thanks{Part of the work was done while Fabien Dufoulon was a postdoctoral fellow at the University of Houston in Houston, USA. During that time, F. Dufoulon was supported in part by National Science Foundation (NSF) grants CCF-1540512, IIS-1633720, and CCF-1717075 and U.S.-Israel Binational Science Foundation (BSF) grant 2016419.} \\ Lancaster University \\ 
\email{f.dufoulon@lancaster.ac.uk}
\and%
Michael Moorman\thanks{M. Moorman was supported in part by  NSF IIS-1633720 REU Supplement grant.} \\ University of Houston \\ \email{mmoorman@cougarnet.uh.edu}%
\and%
William K. {Moses Jr.}\thanks{Part of the work was done while William K. Moses Jr. was a postdoctoral fellow at the University of Houston in Houston, USA. W. K. Moses Jr. was supported in part by NSF grants CCF-1540512, IIS-1633720, and CCF-1717075 and BSF grant 2016419.} \\ Durham University \\ \email{wkmjr3@gmail.com}%
\and%
Gopal Pandurangan\thanks{G. Pandurangan was supported in part by NSF grants CCF-1540512, IIS-1633720, and CCF-1717075 and BSF grant 2016419.} \\ University of Houston \\ \email{gopal@cs.uh.edu}%
}

\begin{document}
\maketitle

\begin{abstract}
\input{abstract.tex}    
\end{abstract}

\section{Introduction}\label{sec:introduction}
\input{introduction}

\section{Preliminaries}\label{sec:preliminaries}
\input{preliminaries}

\section{Our Primitives}\label{sec:primitives}
\input{primitives}

\section{Overlay Construction Protocol}\label{sec:overlayConstruction}
\input{overlayConstruction}

\section{Experimental Results}\label{sec:experiment}
\input{experiment}

\section{Conclusion}\label{sec:conclusion}
\input{conclusion}

\bibliographystyle{plainurl}
\bibliography{reference,experiment}


\appendix

\section{Details on Experimental Results}\label{app:appendixExperimental}
\input{experiment-appendix}

\end{document}

%% file: abstract.tex
We focus on the well-studied problem of distributed overlay network construction.
We consider a synchronous {\em gossip-based} communication model where in each round a node can send a message of small size  to another node
whose identifier it knows. 
The network is assumed to be {\em reconfigurable}, i.e., a node can add new connections (edges) to other nodes
whose identifier it knows or drop existing connections.
 Each node
initially has only knowledge of its own identifier  and the identifiers of its neighbors.
The overlay construction problem is, given an arbitrary (connected) graph, to reconfigure it to obtain a bounded-degree {\em expander} graph as efficiently as possible. The overlay construction problem is relevant to building real-world peer-to-peer 
network topologies that have desirable properties such as low diameter, high conductance, robustness to adversarial deletions, etc.

Our  main result is that we show that starting from any {\em arbitrary} (connected) graph $G$ on $n$ nodes and $m$ edges, we can construct
an overlay network that is a constant-degree {\em expander} in $\polylog{n}$ rounds  using only $\tilde{O}(n)$ messages.\footnote{The notation $\tilde{O}$ hides a $\polylog{n}$ multiplicative factor.}
Our time and message bounds are {\em both} essentially optimal (up to polylogarithmic factors). Our distributed overlay construction protocol
is very lightweight as it uses gossip (each node communicates with only one  neighbor in each round) 
and also scalable as it uses only $\tilde{O}(n)$ messages, which is
{\em sublinear} in  $m$ (even when $m$ is moderately dense). To the best of our knowledge, this is the first
result that achieves overlay network construction in $\polylog{n}$ rounds and $o(m)$ messages. Our protocol
uses graph sketches in a novel way to construct an expander overlay that is both time and communication efficient.

A consequence of  our  overlay construction protocol is that  distributed computation can be performed very efficiently
in this model. In particular,  a wide range of fundamental tasks such as broadcast, leader election,  and minimum spanning tree (MST) construction  can be accomplished in $\polylog{n}$ rounds and $\tilde{O}(n)$ message complexity in any graph.

%% file: introduction.tex
\subsection{Background and Prior Work}
Many of today's large-scale distributed systems in the Internet are peer-to-peer (P2P) or overlay networks.
In such networks, the (direct) connections between nodes  can
be considered as {\em virtual (or logical)} connections as they make use of the physical connections of the underlying Internet.  Furthermore, in these networks,  a node can communicate with another node if it knows the IP address of the other node, and 
can also (potentially) establish a connection (link) to it. 
Thus, the network can {\em reconfigure} itself, choosing which connections to add and which to drop.

In this paper, we consider  the well-studied problem of constructing an efficient
overlay topology in a distributed fashion in a reconfigurable network.
Overlay construction is particularly important in modern P2P networks,
which depend significantly on topological properties to ensure
efficient performance. In fact, over the last two decades, several theoretical works~\cite{PRU01,LS03,mihail-p2p,CDG07-soda,CDG07-soda,JP12,focs2015} have focused on building  P2P networks
with various desirable properties such as high conductance,
low diameter, and robustness to a large number of adversarial deletions.
The high-level idea  in all these works is  to distributively build
a (bounded-degree) \emph{random graph}   topology which guarantees the above  properties. This idea exploits the fact
that a random graph is an {\em expander} with high probability and hence has all the above desirable properties~\cite{Hoory,MUbook}.\footnote{Throughout, ``with high probability (w.h.p.)'' means
with probability at least $1-1/n^c$ for some constant $c \geq 1$; $n$ is the network size.} Indeed, random  graphs have  been used extensively to model P2P networks (see e.g.,~\cite{LS03, PRU01, mihail-p2p, CDG07-soda, Mahlmann_2006, Augustine_2012, Augustine_2013_PODC, byz-leader, Augustine_2013_SPAA}).
It should also be noted that the random connectivity topology is  widely
deployed in many P2P  systems today, including those that implement blockchains and cryptocurrencies, e.g., Bitcoin~\cite{perigee}.

Several prior works \cite{AACWY05,GPRT20,GHSW21} have addressed the problem of constructing an expander topology starting
from an {\em arbitrary} topology network. One of the earliest works
is that of Angluin et al.~\cite{AACWY05} who showed how one can transform an arbitrary connected graph $G$ on $n$ nodes
and $m$ edges
into a binary search tree of depth $O(\log n)$ in  $O(d+\log n)$ rounds and $O(n(d + \log n))$ messages, where $d$ is the maximum degree.
(Their model is similar to ours, where a node can only send a message to a single neighbor
per round.) It can be shown that an $O(\log n)$-depth binary tree can be transformed into many other desirable topologies (such
as an expander, butterfly, or hypercube). 
The work of Gilbert et al.~\cite{GPRT20} presented a distributed protocol  that when given any (connected) network topology having $n$
nodes and $m$ edges will transform it in to a given (desired) target topology such as an expander, hypercube,  or Chord, with high probability. This protocol  incurred    $O(\polylog{n})$ rounds and exchanged messages
of only  small size ($O(\log n)$ bits) per communication link per round, and had a total message complexity
of $\tilde{\Theta}(m)$.\footnote{The notation $\tilde{O}$ hides a $\polylog{n}$ multiplicative factor.}
Note that while the protocol of Gilbert et al. has a better time complexity compared to Angluin et al., when $d$ is large, the protocol of Gilbert et al. is not gossip-based  (i.e., a node can send messages
to all its neighbors in one round, even if its degree is large), unlike that of Angluin et al.

The recent work of G{\"o}tte et al. \cite{GHSW21} presented an overlay construction algorithm that when given
an arbitrary (connected) graph, transforms the graph into a {\em well-formed tree}, i.e., a rooted tree of constant 
degree and $O(\log n)$ diameter. In particular, their protocol first constructs an $O(\log n)$-degree  expander (a well-formed tree can be obtained from the expander by known techniques).   The protocol  takes $O(\log n)$ rounds
which is asymptotically time-optimal, since $\Omega(\log n)$ rounds is a lower bound for constructing a well-formed tree
or a constant-degree expander from an arbitrary graph~\cite{GHSW21}.
However, their protocol takes $\tilde{\Theta}(m)$ messages ($m$ is the number of edges in the starting graph) as each node needs to send $d\log n$ messages in a round
where $d$ is the initial (maximum) degree. The novelty of their protocol is the repeated use of short (constant length) random walks to increase the graph conductance. 

We note that all the above overlay construction protocols, while being fast, i.e., taking only $O(\polylog{n})$ rounds, use $\tilde{\Theta}(m)$
messages, i.e., linear in the number of edges of the initial graph. An important question is whether one can design
overlay construction protocols that are significantly communication-efficient, i.e., taking $o(m)$ messages
or even $\tilde{O}(n)$ messages. In fact, the  work of G{\"o}tte et al.~\cite{GHSW21} raised the question
of whether it is possible to obtain a fast overlay construction protocol that is also communication-efficient.

In this paper, we answer the above question and present the first  overlay construction protocol that
is both time- and communication-efficient: given 
an arbitrary connected graph on $n$ nodes and $m$ edges, the protocol constructs a constant-degree expander in $\polylog{n}$
rounds using only $\tilde{O}(n)$ messages (regardless of the value of $m$). We note that our protocol uses only messages of small size ($O(\polylog{n})$ bits). Furthermore, it uses {\em gossip-based communication} which is fully-decentralized and lightweight. Hence, it inherits the usual advantages  of gossip-based protocols (e.g., see \cite{Kempe02,demers,karp-gossip}) such as robustness,  no single point of action, etc.\footnote{
For this reason, we avoid more centralized methods such as building and using a BFS tree for aggregation, etc., in favor of fully-decentralized gossip-based aggregation protocols~\cite{Kempe02,demers}.}

\subsection{Model}
\label{subsec:model}

Before we formally state our main result, we discuss our model which is similar to that used in previous
work on overlay network construction (see e.g.,~\cite{AACWY05}).

We assume that we are given a connected {\em arbitrary} graph $G = (V,E)$ as input. Let $|V| = n$ and $|E| = m$. Each node has a unique ID (identifier) taken from the range $[1,N]$, where $N = n^c$ for some positive constant $c \geq 1$. Thus, each node
ID can be represented using $O(\log n)$ bits. 

We assume that the communication links are \emph{reconfigurable}: if a node $u$ knows about the ID of some node $v$, then $u$ can establish or drop a link to $v$.\footnote{Strictly speaking, it takes a successful handshake between $u$ and $v$ to establish or drop a bidirectional link.
For simplicity, and since it does not change the asymptotic bounds of our results, we assume that these connections happen instantaneously.} Also, as is standard in overlay (and P2P) networks, a node can  communicate with another node if it knows the identity  of the other node.

The computation proceeds in {\em synchronous} rounds and  the communication topology produced by the execution evolves as a sequence of graphs $G=G_1,G_2,\dots$, where $G_r=(V,E_r)$ corresponds to the network at the beginning of round $r$.  The graph $G_1$ is the initial configuration, and determines the initial knowledge of nodes which is restricted to only knowledge of their own ID and the IDs of their (respective) neighbors. This is a standard model in distributed computing, called  the {\em Knowledge-Till-Radius 1 $(KT1)$} model. Note that in overlay and P2P networks, this model is the natural model, as every node knows the identity (IP address) of itself and 
of its neighbors.  Nodes initially have no knowledge of any other nodes (other than their respective neighbors) or any global knowledge including the initial
topology.  However, we assume that nodes know an upper bound on $n$ (in fact, just a constant factor upper bound of
$\log n$ is sufficient).

We assume a {\em gossip-based} communication model which is very lightweight, where in a round,  a node can send a small-sized message (of size $O(\polylog{n})$ bits) to only
one of its neighbors.  Thus,  a round $r$ consists of the following three steps: (i) each node contacts one neighbor, (ii) each node sends a $O(\polylog{n})$ bit message to the contacted neighbor and receives a $O(\polylog{n})$ bit reply, and (iii)
  after all messages in transit have been received, $u$ performs some local computation possibly including changes to its communication links, resulting in $G_{r+1}$. 
We call this model the \textsf{P2P-GOSSIP} model. Note that although the gossip model allows each node to send a message to only one neighbor per round, one can easily simulate sending messages to $d$ neighbors in a round, by performing gossip for $d$ rounds,
thus blowing up the number of rounds by a factor of $d$. On the other hand, in the standard CONGEST model, 
a node can send a message (of small size, say $O(\log n)$ bits) to all its neighbors in a round. In addition, in the LOCAL model, the message
size is unbounded.

\subsection{Our Contributions}
\label{subsec:contribution}

\noindent {\bf Main Result.} Our main contribution is a distributed protocol for overlay network construction that given an {\em arbitrary} connected  graph, constructs an overlay graph whose topology is a constant-degree {\em expander}. Informally,  an expander  is a graph that has {\em constant} conductance (see definition in Section \ref{subsec:definitions}). 
As mentioned earlier, an expander graph has very desirable properties: low diameter ($O(\log n)$), high conductance, fast mixing of random walks   (i.e., a random walk reaches stationary
distribution  in $O(\log n)$ rounds, which is useful for fast random sampling), and robustness to large adversarial deletions (deleting
even a constant fraction of nodes leaves a giant component of size $\Theta(n)$ that is also an expander)~\cite{Hoory,bagchi,sigact,byz-leader}.

The protocol takes $\polylog{n}$ rounds and uses only $\tilde{O}(n)$ messages.\footnote{We have not chosen to optimize the log factors
in our protocol, as this was not the primary focus. As it is, our protocol takes $O(\log^5n)$ rounds and this can be improved.}
These time and message bounds are {\em both} essentially optimal (up to polylogarithmic factors), since it is easy 
to show that 
$\Omega(\log n)$ is a lower bound on time \cite{GHSW21} and $\Omega(n)$ is a lower bound
on the number of connections that need to be added/deleted (and hence the number of messages).\footnote{
Consider a dumbbell graph consisting of two cliques joined by a single edge as the starting graph. To convert
this graph to a constant degree expander,  at least $\Theta(n)$ edges have to be added between the cliques, and the cliques
themselves have to be sparsified by dropping all but a constant number of random edges.}

Our distributed overlay construction protocol (Section \ref{sec:overlayConstruction})
is very lightweight and scalable, as it uses gossip-based communication (each node communicates with only one
neighbor per round). 
 To the best of our knowledge, this is the first
result that achieves overlay network construction in $O(\polylog{n})$ rounds and $o(m)$ messages, i.e., sublinear
in $m$, the number of edges of the initial graph. All prior protocols took at least $\tilde{O}(m)$ messages in general while taking $\polylog{n}$ rounds. We note that once an expander topology is constructed, several other well-known  topologies such as hypercube,
butterfly, binary tree,  etc.\ can be constructed~\cite{GHSW21,AACWY05}.

We simulated our protocol to study its performance  in several types of graphs (Section \ref{sec:experiment}). The results validate
the theory and shows that in all the different types of graphs, the conductance increases significantly
to essentially the best possible. The algorithm also finishes fast, i.e., in a few phases.

\medskip

\noindent {\bf Implications.} A  consequence of  our  overlay construction protocol, is that  distributed computation can be performed very efficiently in the P2P-GOSSIP model. In particular,   a wide range of fundamental tasks such as broadcast, leader election,   and  spanning tree (ST) construction can be accomplished in $\polylog{n}$ round complexity and $O(n\polylog{n})$ message complexity in any graph. This follows  because one can first construct a constant-degree expander using the overlay construction protocol and then do the above tasks on the expander graph (which has $O(\log n)$ diameter and $O(n)$ edges) 
in $O(\log n)$ rounds and $\tilde{O}(n)$ messages by just simulating
standard CONGEST model algorithms in gossip~\cite{Peleg_2000_Book,Kutten_2015_JACM}.\footnote{Since the graph is of constant degree, one can easily simulate a round of the CONGEST model,  where a node sends a message to all its neighbors, by performing gossip for a constant number of rounds.}

Furthermore, one can also show that the   minimum spanning tree (MST) problem can be solved
very efficiently. (Note that in the MST problem, we are
 given an arbitrary (connected) undirected graph $G$ with edge weights,  and the goal is to find the MST of $G$.)
 We will outline how this can be accomplished in $\polylog{n}$ rounds and $O(n\polylog{n})$ messages which is
 a consequence of this work and prior works.
 First, using our expander overlay protocol, we add 
a (constant-degree) expander overlay  on  $G$ (i.e., the expander edges are added in addition to the edges of $G$).
The expander edges will be used for efficient communication in $G$. For this, we 
convert the expander (that is not addressable) into an   hypercube (that is addressable) which allows
efficient routing between any two nodes in $O(\log n)$ rounds and $O(n\log{n})$ messages.
This conversion  can be accomplished using the techniques of \cite{AACWY05,gmyr-hybrid,aspneswu} (see also \cite{GHSW21}) or the protocol of \cite{spaa22}.\footnote{
The techniques of \cite{AACWY05,aspneswu,GHSW21}  construct a  well-formed tree which can then be transformed into
an hypercube. The well-formed tree can be easily constructed from an expander \cite{GHSW21}. Alternatively,
one can use the protocol of \cite{spaa22} which  gives a fully-decentralized and robust protocol  using random walks (which can be simulated using gossip) to construct an addressable  hypercube from a constant-degree expander.} All these protocols  take $\polylog{n}$ rounds
and $\tilde{O}(n)$ messages to convert a constant-degree expander into an hypercube. 
Using the addressable hypercube overlay on top of $G$, we can 
 efficiently implement the Gallagher-Humblet-Spira (GHS) algorithm \cite{GHS83}  as shown in Chatterjee et al.~\cite{CPP20} to compute the MST of $G$ in $\polylog{n}$ rounds and $\tilde{O}(n)$ message complexity.\footnote{We note that Chatterjee et al. uses the permutation routing algorithm of Ghaffari and Li~\cite{jasonli} (also see \cite{GhaffariKS17}) which takes
 $2^{O(\sqrt{\log n})}$ rounds on an expander (which is not addressable). In contrast, exploiting
 the reconfigurability of the  P2P model, we can {\em reconfigure} the expander into an hypercube which allows
 permutation routing to be accomplished in $O(\log n)$ rounds  using the standard Valiant routing~\cite{valiant}.} 

We note that the above results are a significant improvement in the round complexity of solving the above
fundamental problems in the standard $KT1$ model. King, Kutten, and Thorup  \cite{KKT15} showed
that all these problems can be solved using $\tilde{O}(n)$ messages (regardless of the value of $m$, the number
of edges of the graph), but this takes $\tilde{O}(n)$ rounds in the standard CONGEST model.
Obtaining  sublinear, i.e., $\tilde{O}(n)$  messages, in time that is significantly faster than $O(n)$ is an open problem in the CONGEST model (see also \cite{MashreghiK17,gmyr,KKT15}). In contrast, we show that in the P2P-GOSSIP model, one can solve these problems
in sublinear ($\tilde{O}(n)$) messages and $\polylog{n}$ time.

Another implication is to the complexity of the  GOSSIP model, see e.g., \cite{CS12, CHKM12} and the references therein. The works of Censor-Hillel et al.~\cite{CHKM12} studied the complexity of
distributed computation in the GOSSIP model --- where a node may only initiate contact with a single neighbor in each round, but  {\em unbounded} messages sizes are allowed (unlike in our P2P-GOSSIP model which uses small message sizes)
--- in comparison to the more standard model of distributed computation, namely, the much less restrictive LOCAL  model, where a node may simultaneously communicate with all of its neighbors in a single round (also message sizes can be unbounded). This work studied the complexity
of the {\em information dissemination} problem in which each node has some (possibly unique) data item and each node is required to collect all the data items from all nodes.  They gave an algorithm that solves the information dissemination problem in at most $O(D+\polylog {n})$ rounds in a network of size $n$ and diameter $D$. This is at most an additive polylogarithmic factor from the trivial lower bound of $D$, which applies even in the LOCAL model. In fact, they prove that any algorithm that requires $T$ rounds in the LOCAL model can be simulated in $O(T+\polylog{n})$ rounds in the GOSSIP model, showing that GOSSIP and LOCAL models are essentially equivalent (up to polylogarithmic factors).
Our work shows that in the P2P-GOSSIP model, if we allow unbounded message sizes, one can solve information dissemination
in $O(\polylog{n})$ rounds in a straightforward way by first constructing a constant-degree expander (that has diameter
$O(\log n)$) and then doing gossip on the expander~\cite{CS12}.

\medskip

\noindent {\bf High-level Overview and Technical Contributions.}
Our overlay construction protocol (Section \ref{sec:overlayConstruction}) uses a combination of several techniques in a non-trivial way
to  construct an expander overlay in $\tilde{O}(n)$ messages and $\polylog{n}$ time. 

Our expander overlay construction protocol is conceptually simple and is similar  to the
classic GHS algorithm~\cite{GHS83} and consists of several phases. In the first phase, we start with  $n$ clusters, each corresponding to a node.  In a phase, adjacent clusters
are merged and the protocol maintains the invariant
that each cluster is a {\em constant-degree expander} at the end of each phase.  

We now describe a phase of the protocol which consists of three major steps. By the maintained invariant, we can assume
that all the clusters are constant-degree expanders (this is trivially true in the first phase, since each cluster is
a singleton node). As in GHS, one has to quickly find outgoing edges from each cluster. Except for the first phase, it is non-trivial to accomplish this using gossip in $\polylog{n}$ rounds and using
$O(n\polylog{n})$ messages per phase.
The first major step  in a phase is to efficiently aggregate sketches from each cluster. 
Informally, a sketch of a node is a short representation of its adjacency list, i.e., using $O(\polylog{n})$ bits (see Section \ref{subsec:graph-sketch}). An important property of these graph sketches is that the sketches of all nodes in a cluster can be aggregated in $O(\polylog{n})$ bits and an outgoing edge of the cluster can be found with high probability from
the aggregated sketch.
A main technical contribution, that can be of independent interest, is showing how one can use  graph sketches \cite{AhnGM12a,AhnGM12b}
 to efficiently sample an outgoing edge using gossip in an expander. We show that we can adapt the PUSH-SUM  gossip protocol
 of \cite{Kempe02} to efficiently aggregate sketches in an accurate manner in all nodes of a cluster (Section \ref{subsec:push-sum}). 
 
 In the second step, each node uses the aggregated sketch to sample an outgoing edge.
 The clusters along with their respective outgoing edges induce a disjoint set of connected components.
 Each such component has to be converted into a constant degree expander.
To accomplish this, we use the protocol of G{\"o}tte et al. \cite{GHSW21} that takes an arbitrary  {\em constant-degree} graph
and converts it into an expander that has $O(\log n)$ degree.\footnote{Note that the protocol of \cite{GHSW21} works on graphs of somewhat higher degree, say $O(\log n)$,
but the maximum degree needs to be small to get the performance guarantees as claimed.} One cannot directly invoke the protocol of \cite{GHSW21}
on each connected component, since the degree of a node may not be constant (multiple outgoing edges can go to a single node).   In the second step, to satisfy the degree bound needed for the protocol of \cite{GHSW21} (invoked in step 3) we do a {\em degree reduction} to keep the degrees of all nodes constant.

In the third step, we run the protocol of \cite{GHSW21} which uses random walks (and can be simulated in the gossip model
with an $O(\log n)$ factor slow down in the number of rounds) to convert each connected component into
an $O(\log n)$-degree expander in $\polylog{n}$ rounds of gossip. We then run an efficient distributed protocol for  reducing
the degree of the expander to $O(1)$, which may also be of independent interest (Section \ref{subsec:degreeReductionExpander}).
This is essential in maintaining the invariant of the next phase (without this degree reduction, one can show
that nodes' degrees can grow large from phase to phase).

We show that each phase  reduces the number of clusters by a constant factor and hence the total number of phases is $O(\log n)$. At the end, the whole graph becomes one cluster, which will be a constant degree expander.

\subsection{Additional Related Work and Comparison}
\label{subsec:related}
There have been several prior works on P2P and reconfigurable networks (e.g., \cite{Pandurangan_2014,dinitz,LS03,focs2015}) that assume 
  an expander graph to start with, and then faults occur (insertions or deletions of nodes due to churn or other dynamic changes), even repeatedly, and the goal is to maintain the expander.  
  One can view the current paper (as well as prior works discussed earlier \cite{GHSW21,AACWY05,GPRT20})  as a preprocessing step that starts with an arbitrary graph and converts it to an expander.  Subsequently, one may use the protocols from these works to maintain the expander. For future work, it will be useful to  extend our protocol to a dynamic setting where churn is present. In this context, one  can see the current work as a first step towards obtaining a protocol that builds and maintains an expander overlay in {\em dynamic} P2P networks that suffer from churn and where the churn can result in  (intermediate) arbitrary topologies that are far from expanders (but still connected, say). In that case, a protocol like ours will be useful to reconstruct an expander.

 Finally, we point out that there has been work on other models that are similar, yet different, compared to the
 P2P-GOSSIP model. Two notable examples include the node-congested clique model \cite{ncc} and the hybrid model \cite{gmyr-hybrid}.
 In the node-congested clique model, each node is constrained to send only a small amount of messages of small size per round, say
 $O(\log n)$ messages each of size $O(\log n)$ bits. In the hybrid model, two modes of communication are assumed:
 the  LOCAL (or CONGEST) model in the input graph $G$, where a node
 can send unlimited (or small-sized) messages on the (local) edges of $G$ and a node-congested clique model where a node
 can send $O(\log n)$-size messages to $O(\log n)$ other nodes. Note that instead of a clique model (where a node
 can potentially communicate with any other node), one can also assume an $O(\log n)$-degree expander (or another sparse degree structure such as a tree) built on top of $G$ and nodes can communicate on these using small-sized messages.
 This variant of the hybrid model is assumed in \cite{GHSW21}. One main difference between the P2P-GOSSIP model
 and the node-congested clique model is that in the former, while a node sends a messages to 
 only one neighbor, a node can receive as many messages as its degree size (e.g., the center node in a star).
 Furthermore, in the node-congested clique model, it is assumed that each nodes knows the IDs of all other nodes which makes
 routing trivial. Whereas in the P2P-GOSSIP model, a node knows only the IDs of its  neighbors.
 Efficient algorithms for MST (and other problems) that take polylogarithmic rounds are known in the node-congested clique
 model. In the hybrid model variant of \cite{GHSW21}, logarithmic round algorithms have been given for several fundamental problems such as spanning trees, MIS, etc. However, it is left open whether an efficient algorithm can be designed for MST in the hybrid model \cite{GHSW21}.
 

%% file: preliminaries.tex
\subsection{Graph Definitions}
\label{subsec:definitions}

For any graph $G=(V,E)$, its (maximal) connected components are called \emph{clusters}. Moreover, for any node $v \in V$, the neighbors of $v$ are denoted by $N(v) = \{u \mid (u,v) \in E\}$ and the degree of $v$ by $d(v) = |N(v)|$. The \emph{volume} of any subset $S \subseteq V$ is defined as $vol(S) = \sum_{v \in S} d(v)$, and the \emph{edge cut} of $S$ as $E(S,V \setminus S) = \{(u,v) \mid (u,v) \in E, u \in S, v \notin S \}$.
The \emph{conductance} $\Phi(S)$ of any subset $S \subseteq V$ is defined as $\Phi(S) = |E(S,V \setminus S)| /\min\{vol(S), vol(V \setminus S)\} $. The conductance of graph $G$ is defined as $\Phi = \min_{S \subseteq V, S \neq \emptyset} \Phi(S)$. 
Finally, for any edge set $E' \subseteq V^2$, we say that $E'$ \emph{generates} the graph $G(E') = (V,E')$ is called the \emph{graph generated}. Note that $E'$ may include edges in $V^2 \setminus E$ (i.e., peer-to-peer edges when $G$ is the communication graph) and thus $G(E')$ is not the graph induced by $E'$.

A family of graphs $G_n$ on $n$ nodes is an \emph{expander family} if,
    for some constant $\alpha$ with $0 < \alpha < 1$, the conductance
    $\phi_n = \phi(G_n)$ satisfies $\phi_n \geq \alpha$ for all
    $n \geq n_0$ for some $n_0 \in \mathbb{N}$.

\subsection{Graph Sketches}\label{subsec:graph-sketch}

Consider an arbitrary set of nodes $V$, such that all nodes of $V$ have unique IDs in $[1,N]$, where $N$ is known to all nodes. 
We define, for any graph $G = (V,E)$ and node $v \in V$, the \emph{incidence vector} $\incVect(v) \in \mathbb{R}^{\binom{N}{2}}$ whose entries correspond to all possible choices of two IDs in $N$.  An entry in $\incVect(v)$  corresponding to the possible edge between $u$ with $id_u\in N$ and $v$ is $0$ if $(u,v) \notin E$, $1$ if $(u,v) \in E$ and $id_u > id_v$, and $-1$ otherwise (i.e., $(u,v) \in E$ and $id_u < id_v$). Entries in $\incVect(v)$ that correspond to a possible edge not including $v$ (e.g., $(u,w)$) have value $0$. 
Naturally, one can extend this definition to any node subset $S \subseteq V$: more precisely, $\incVect(S) = \sum_{v \in S} \incVect(v)$. Note that by linearity, the non-zero indices of $\incVect(S)$ indicate exactly which edges are in the cut of $S$ with respect to $G$, that is, in $E_G(S, V \setminus S)$. 

One may use the incidence vector of some node set $S$ to sample an outgoing edge from $S$, if one exists, uniformly at random from all such outgoing edges, i.e., sample a non-zero entry in $\incVect(S)$ uniformly at random. However, in a distributed setting, to compute the incidence vector on a set of nodes $S$, one would need to aggregate the incidence vectors of the nodes belonging to $S$.   
This is problematic since incidence vectors are exponentially larger (recall that they have size $\binom{N}{2}$) than our $O(\polylog N)$ message size. 

Fortunately, we can use a \emph{linear sketch} \cite{AhnGM12a} --- a well-chosen linear function from $\mathbb{R}^{\binom{N}{2}}$ to $\mathbb{R}^k$ --- or in other words, a \emph{graph sketching matrix}  --- a well-chosen $k \times \binom{N}{2}$ size matrix $M_G$ --- to compress these vectors $\left( \incVect(v) \right)_{v \in V}$ into smaller (sketch) vectors $\left( \sketchVect(v) \right)_{v \in V}$ of size $k = O(\polylog{\binom{N}{2}})$; more concretely, $M_G \cdot \incVect(v) = \sketchVect(v)$. 
Moreover, although this compression necessarily loses some information, it has two major advantages. First, it is possible to sample (almost) uniformly at random one of the non-zero indices of $\incVect(v) \in \mathbb{R}^{\binom{N}{2}}$ by performing operations on $\sketchVect(v)$ only (albeit with some small failure probability). More concretely, for any graph sketching matrix $M_G$ and for any subset $S \subseteq V$, there exists a sampling function $f_G$ that takes the sketch vector $\sketchVect(S)$ as input and outputs an edge chosen uniformly at random in $E_G(S,V \setminus S)$. 
Second, the linearity of the graph sketching matrix allows us to compute $\sketchVect(S)$ without computing $\incVect(S)$, but instead by 
computing the sketch vectors $\left( \sketchVect(v) \right)_{v \in S}$ and summing them. 

In summary, for any subset $S \subseteq V$, we can sample an edge chosen uniformly at random in $E_G(S,V \setminus S)$ by (i) having nodes agree on $\Theta(\log^2 n)$ true random bits that they can use to locally compute a common graph sketching matrix $M_G$ with polynomially bounded integer entries~\cite{JST11, PRS18}, (ii) aggregating the sketch vectors of all nodes in $S$, and (iii) applying the sampling function $f_G$ on the aggregate vector $\sketchVect(S)$. Note that these steps require messages of small $O(\polylog n)$ size only. 
A more formal statement is given below, which can be obtained straightforwardly by adapting known results~\cite{JST11, AhnGM12a, PRS18}.

\begin{lemma}\label{lem:graph-sketch-works}
    For any upper bound $N$ on the ID range and constant $0 < \delta < 1$, there exist a graph sketching matrix $M_G$ (with entries polynomially bounded in $N$) and a sampling function $f_G$ such that for any node subset $S \subset V$, the aggregate sketch vector $\sketchVect(S) = \sum_{u \in S} \sketchVect(u)$ can be represented using $O(\polylog n)$ bits, and $f_G(\sketchVect(S))$ samples a (uniformly) random edge in $E_G(S,V \setminus S)$ with probability $1 - \delta$.
\end{lemma}

\subsection{Creating Expanders from Bounded Degree Graphs}
\label{subsec:createExpander}

To make our work self-contained, we briefly describe Procedure \CreateExpander{} --- the overlay construction algorithm from \cite{GHSW21} --- and its guarantees here. For any constant $\Phi \in (0,1/2]$, integer $d \geq 1$ and $d$-bounded degree graph $G = (V,E)$, Procedure \CreateExpander{} first performs some preprocessing on the graph $G$ to transform it into an $O(\log n)$-regular benign graph $H_0$. (A formal definition of benign graphs is given below, but roughly speaking, these are graphs on which random walks have some good properties.) After which, starting from the edges of $H_0$, Procedure \CreateExpander{} computes in each phase a new set of edges generating a graph with twice as much conductance --- up to 1/2. After some $O(\log n)$ phases, Procedure \CreateExpander{} terminates and outputs an $O(\log n)$-regular expander graph with conductance $\Phi$.

\begin{lemma}[\cite{GHSW21} with an $O(\log n)$ overhead in the \textsf{P2P-GOSSIP} model]
\label{lem:GotteResult}
For any constant $\Phi \in (0,1/2]$, integer $d \geq 1$ and $d$-bounded degree graph $G = (V,E)$, Procedure \CreateExpander{} uses $O(\log^2 n)$ rounds and $O(n \log^2 n)$ messages to output an $O(\log n)$-regular graph $H'=(V,E_H')$ with conductance $\Phi$.
\end{lemma}

Next, we provide more details about the phases of Procedure \CreateExpander{}. First, let $\Delta_H = \Delta_H(d) = \Theta(\log n)$ and $\Lambda = \Theta(\log n)$ be two parameters chosen via the analysis. We call $H = (V,E_H)$, a multi-graph composed of peer-to-peer (multi-)edges, a \emph{$(\Delta_H,\Lambda)$-benign graph} if it holds that (i) $H$ is $\Delta_H$-regular, that is, every node of $H$ has $\Delta_H$ incident (multi-)edges, (ii) $H$ is lazy, that is, every node of $H$ has at least $\Delta_H/2$ self-loops, and (iii) the minimum cut of $H$ is at least $\Lambda$. 
Then, Procedure \IncreaseExpansion{} --- implementing one phase of Procedure \CreateExpander{} --- is run on a $(\Delta_H,\Lambda)$-benign graph and outputs a $(\Delta_H,\Lambda)$-benign graph with better conductance. More precisely, in Procedure \IncreaseExpansion{}, each node initiates $\Delta_H/8$ random walks of length $\ell = \Theta(1)$. These random walks take $\Theta(\log n)$ rounds to terminate (since nodes may only send 1 message per round in the \textsf{P2P-GOSSIP} model). After these $O(\log n)$ rounds, each node generates, for each token it holds, an edge to the token's originating node. If a node holds more than $3\Delta_H/8$ tokens, then that node randomly chooses $3\Delta_H/8$ tokens without replacement and creates edges accordingly. Finally, each node adds self-loops until it has $\Delta_H$ incident edges.

\begin{lemma}[\cite{GHSW21}]
\label{lem:expansionIncrease}
For any $(\Delta_H, \Lambda)$-benign graph $H = (V,E_H)$ with conductance $\Phi$, Procedure \IncreaseExpansion{} uses $O(\log n)$ rounds and $O(n \log n)$ messages to output a $(\Delta_H, \Lambda)$-benign graph $H'=(V,E_H')$ with conductance $\Phi' \geq \min\{ (\Phi \sqrt{\ell})/640, 1/2\}$. For $\ell \geq (2 \cdot 640)^2$, it holds that $\Phi' \geq \min\{ 2 \Phi, 1/2\}$.
\end{lemma}

%% file: primitives.tex
In this section, we develop two primitives (that can be of independent interest) that we subsequently use in our main algorithm.  The first primitive, in Section~\ref{subsec:degreeReductionExpander}, allows us to modify an $O(\log n)$-regular expander graph into an $O(1)$-bounded degree expander. The second primitive, in Section~\ref{subsec:push-sum}, allows one to construct an aggregate sketch vector in an $O(1)$-bounded degree expander using a gossip-style approach, assuming that each node initially contains its own sketch vector.

\subsection{Degree Reduction for Expanders}
\label{subsec:degreeReductionExpander}

Consider a high conductance $O(\log n)$-regular expander graph $G$. 
We give a procedure, called Procedure \ExpanderDegreeReduction{}, that sparsifies this expander graph into a $\delta$-bounded degree expander graph $G_{\delta}$ with any desired conductance $\Phi \in (0,1/10]$, in $O(\log^3 n)$ rounds and $O(n \log^3 n)$ messages, where $\delta = O(1)$ is some integer determined by the analysis. 

Initially, each node generates $c$ active tokens with its ID, where $c$ is a positive integer such that $c < \delta/10$. These tokens will be used to generate the edges of the procedure's resulting graph $G_{\delta}$. The algorithm works in $O(\log n)$ phases, where the precise value is determined by the analysis. In each phase, active tokens take random walks of length $\ell = O(\log n)$ per phase --- where $\ell$ is larger than the mixing time on $G$. (More concretely, each node that holds an active token sends the token to a neighbor chosen uniformly at random.\footnote{We show later that sending these tokens in our model incurs an overhead of $O(\log n)$.}) At the end of each phase, any token that ends at some node with at most $\delta$ tokens becomes inactive and remains at that node for the remainder of the algorithm. (If during some phase too many tokens end at a node, such that it holds strictly more than $\delta$ tokens, then for simplicity all these tokens remain active.) Moreover, for each token that becomes inactive, the node holding it creates a temporary (1-round) edge to inform its source (of the inactive token). A node is said to be \emph{satisfied} once all of its generated tokens are inactive. Once all nodes are satisfied, each satisfied node creates one edge (in $G_{\delta}$) to each node holding at least one of its inactive tokens.  

\noindent \textbf{Analysis.}
We start with a simple invariant obtained via counting argument (see Lemma \ref{lem:goodNodeInvariant}). With this invariant, we can show that the number of active tokens reduces by half in each phase with constant probability (see Lemma \ref{lem:nonAcceptedTokensDecrease}). As a result, we can show that all nodes become satisfied within $O(\log n)$ phases (see Lemma \ref{lem:satisfiedNodes}), or in other words, before the algorithm terminates. 

\begin{lemma}
\label{lem:goodNodeInvariant}
For any given phase, fix the random walks of all tokens except one. Then, it holds that at least $0.9n$ nodes have strictly less than $\delta$ tokens.
\end{lemma}

\begin{proof}
We use a simple counting argument to show the lemma statement. Assume by contradiction that more than $n/10$ nodes have more than $\delta$ tokens. By the algorithm definition, there can only be $c \cdot n$ tokens in total, and thus $c \cdot n \geq \delta \cdot n / 10$. However, since $c < \delta/10$, $c \cdot n < \delta \cdot n /10 $, which leads to a contradiction.
\end{proof}

\begin{lemma}
\label{lem:nonAcceptedTokensDecrease}
In any given phase, the number of active tokens reduces by half with probability at least $1/2$.
\end{lemma}

\begin{proof}
    Let $k \leq c \cdot n$ be the initial number of active tokens in this phase. These tokens, denoted by $t_1,\ldots,t_k$, each take an $\ell$-length random walk. For any $i \in \{1,\ldots,t\}$, let $\V_i$ be the random variable denoting the node the $i$-th token ends at. Note that $\ell$ is chosen sufficiently greater than the mixing time on $G$, thus the $i$-th token ends at an (almost) uniform random node. More formally, for any node $v \in V$, $\Pr[\V_i = v] \in [1/n - 1/n^a, 1/n + 1/n^a]$ for some constant $a \geq 1$. Then, the following rough upper bound holds: for some constant $\varepsilon \in (0,1)$, for any node $v \in V$, $\Pr[\V_i = v] \leq (1+\varepsilon)/n$. 

    Next, for any token $t_i$, let $R(t_i) = (X_0,\ldots,X_\ell)$ denote the random walk of $t_i$ (i.e., the sequence of nodes visited by token $t_i$) and let $\Ind_i$ indicate that $t_i$ ends the phase as active. We use $\bar{R}(t_i)$ as shorthand for $\left( R(t_1),\ldots,R(t_{i-1}),R(t_{i+1}),\ldots,R(t_k)\right)$, or in other words, to denote all tokens' random walks except that of token $t_i$. Note that if we fix all random walks except that of $t_i$ --- i.e., if we fix $\bar{R}(t_i)$ --- then Lemma \ref{lem:goodNodeInvariant} implies that there exists a set $G(\bar{R}(t_i))$ of nodes, each with strictly less than $\delta$ tokens, and such that $|G(\bar{R}(t_i))| \geq 0.9 n$. 
    If $t_i$ ends up at any nodes of $G(\bar{R}(t_i))$, it becomes inactive. Thus, $\Pr[\Ind_i = 1 \mid \bar{R}(t_i)] \leq \Pr[\V_i \notin G(\bar{R}(t_i)) \mid \bar{R}(t_i)] \leq \Pr[\bigvee_{v \notin G(\bar{R}(t_i))} \V_i = v \mid \bar{R}(t_i)] \leq (1+\varepsilon)/10$, where the last inequality is obtained through union bound. Consequently, $\Pr[\Ind_i = 1] \leq (1+\varepsilon)/10$.
    
    Let the random variable $A = \sum_{i=1}^k \Ind_i$ denote the number of active tokens when the phase ends. By the above inequality, $\E[A] =  \sum_{i=1}^k \Pr[\Ind_i = 1] \leq (1+\varepsilon) k / 10$. Finally, Markov's inequality implies the lemma statement: that is, $\Pr[A \geq 0.5 k] \leq \Pr[A \geq 2 \E[A]] \leq 1/2$. 
\end{proof}

\begin{lemma}
\label{lem:satisfiedNodes}
After $O(\log n)$ phases, all nodes are satisfied.
\end{lemma}

\begin{proof}
    A phase is said to be successful if the number of active tokens reduces by half. By Lemma \ref{lem:nonAcceptedTokensDecrease}, each phase is successful with probability at least $1/2$. A simple application of Chernoff bounds imply that there are at least $\log (c \cdot n)$ successful phases after large enough $O(\log n)$ phases. Consequently, after $O(\log n)$ phases, there remain no active tokens, and hence no unsatisfied nodes.
\end{proof}

Now, we can show the correctness of the primitive, and bound its round and message complexities, in Theorem \ref{thm:expanderDegreeReductionGuarantees}. 

\begin{lemma}
\label{lem:probaInequality}
For any constant $\Phi \in (0,1/10]$ and any two integers $n,s \geq 1$ such that $s \leq n/2$, $\binom{n}{s} \binom{cs}{(1-\Phi) c s} ((1+\varepsilon)s/n)^{(1-\Phi) c s} \leq 1/n^2$ for large enough $n$ and some suitably chosen integer $c \geq 1$ and constant $\epsilon \geq 0$.
\end{lemma}

\begin{proof}
    Let $p^* = \binom{n}{s} \binom{cs}{(1-\Phi) c s} ((1+\varepsilon)s/n)^{(1-\Phi) c s}$. We give two upper bounds: the first for $s = o(n)$ and the second for $s = \kappa n$ for some constant $0 < \kappa \leq 1/2$.
    
    To get the first bound, we use the inequality $\binom{y}{x} \leq (ey/x)^x$, that holds for any integers $y \geq x \geq 1$. Then, $p^* \leq (en/s)^s (e/(1-\Phi))^{(1-\Phi) c s} ((1+\varepsilon) s/n)^{(1-\Phi) c s} = 2^{s(\beta + (1-(1-\Phi)c) (\log n - \log s))}$, where $\beta = \log(e) + (1-\Phi)c \log (e (1+\varepsilon) / (1-\Phi)) $. For large enough $n$, the $(1-(1-\Phi) c) \log n$ factor dominates in the exponent. Thus, 
    it suffices to choose $c$ large enough and it holds that $p^* \leq 1/n^2$ for large enough $n$.
    
    To get the second bound, we need to use a tighter inequality (since $s$ and $n$ are large) to bound the binomial coefficients. More concretely, using Stirling's formula, it holds that $\binom{y}{x} \leq 2^{y H(x/y)}$, for any integers $x, y \geq 0$ and where $H(q) = -q \log(q) - (1-q) \log (1-q)$ is the binary entropy of $q \in (0,1)$. As a result, we get $p^* \leq 2^{n ( H(\kappa) + c \kappa H(1-\Phi) + (1-\Phi) c \kappa \log ((1+\varepsilon) \kappa))}$. 
    First, 
    we take small enough $\varepsilon$ such that for any $\kappa \leq 1/2$, $|\log ((1+\varepsilon) \kappa)| \geq 0.9$. 
    Next, note that 
    $H(1-\Phi) \leq 2 \sqrt{\Phi (1-\Phi)} \leq 0.6$, where the first inequality is a well-known upper bound for binary entropy that holds for any $\Phi \in (0,1)$ and the second one holds because $2 \sqrt{x (1-x)}$ takes value $0.6$ at $x=1/10$ and increases between $x=0$ and $x=1/2$. Then, $H(1-\Phi) \leq |(3/4) (1-\Phi) \log ((1+\varepsilon) \kappa)|$, since the right-hand side is strictly greater than $0.6$ for $\Phi \leq 1/10$.
    Therefore, $p^* \leq 2^{n ( H(\kappa) + (c \kappa/4) (1-\Phi) \log ((1+\varepsilon) \kappa))}$. Next, we take $c$ large enough so that $|(c \kappa/8) (1-\Phi) \log ((1+\varepsilon) \kappa)| \geq 2 |\kappa \log \kappa|$. Then, $p^* \leq 2^{n ( \kappa \log \kappa - (1-\kappa) \log(1-\kappa) + (c \kappa/8) (1-\Phi) \log ((1+\varepsilon) \kappa))}$. Since $\kappa \log \kappa - (1-\kappa) \log(1-\kappa) \leq 0$ for $0 \leq \kappa \leq 1/2$ and the other term in the exponent is negative,  $p^* \leq 2^{n (c \kappa/8) (1-\Phi) \log ((1+\varepsilon) \kappa)}$. Finally, we use again $\log ((1+\varepsilon) \kappa) \leq -0.9$ to obtain $p^* \leq 2^{- b \cdot s }$, where $b = 0.9 c (1-\Phi) / 8$ does not depend on $\kappa$. Since $s = \Omega(n)$, it holds that $p^* \leq 1/n^2$ for large enough $n$. 
\end{proof}

\begin{theorem}
\label{thm:expanderDegreeReductionGuarantees}
For any constant $\Phi \in (0,1/10]$ and for any $O(\log n)$-regular expander graph $G = (V,E)$ with constant conductance, Procedure \ExpanderDegreeReduction{} 
uses $O(\log^3 n)$ rounds and $O(n \log^3 n)$ messages to output an $O(1)$-bounded degree expander graph with conductance $\Phi$ with high probability. 
\end{theorem}

\begin{proof}
    We start with the round complexity. By the algorithm definition, we run $O(\log n)$ phases, and within each phase, tokens take at most $\ell = O(\log n)$ steps. Moreover, a simple randomized analysis shows that each node receives, for any $0 \leq i \leq \ell$, in expectation $O(c) = O(1)$ tokens having taken $i$ steps and per incident edge. Thus, by Chernoff bounds, each node receives, with high probability, at most $O(\log n)$ tokens having taken $i$ steps, for any $0 \leq i \leq \ell$. (A more detailed analysis can be found in the proof of Lemma 3.2 in~\cite{jacm13}.)  Thus, under the condition that tokens having taken less steps have priority to be sent, all tokens take $\ell$ steps within $O(\log^2 n)$ rounds with high probability. As for the message complexity, each node can send at most 1 messages per round, thus the message complexity follows from the round complexity.  
    Next, the resulting graph trivially has bounded degree since by the algorithm definition, each node has up to $c + \delta \leq 2\delta$ incident edges. 

    It remains to show that the resulting graph has constant conductance with high probability. To do so, we give a similar proof to that of Lemma 1 in \cite{focs2015}. To start with, all nodes are satisfied when the algorithm ends, by Lemma \ref{lem:satisfiedNodes}. We consider an arbitrary $S \subset V$ of size $s \leq n/2$ and the following random variables: the edges $\mathcal{E}_1,\ldots,\mathcal{E}_{cs}$ obtained via the algorithm, each corresponding to a now inactive token originating in $S$. After which, let the random variable $\mathcal{L}(S)$ denote the number of edges with both endpoints in $S$. We shall upper bound $\Pr[\mathcal{L}(S) \geq (1-\Phi) c s]$.
    For any integer $i \in \{1,\ldots, c s\}$, let $\Ind_i$ be the indicator random variable indicating whether $\mathcal{E}_i$ has both endpoints in $S$. Since each token is obtained from a random walk of length at least $\ell$, where $\ell$ is greater than the mixing time of $G$, then for any integer $i \in \{1,\ldots, c s\}$, it holds by union bound that $\Pr[\Ind_i = 1] \leq (1+\varepsilon) s/n$. Conditioning on the events $\Ind_1 = 1,\ldots,\Ind_{i-1} = 1$ can only reduce that probability, since nodes (in $S$) can hold a maximum of $\delta$ inactive tokens. Thus, $\Pr[\Ind_i = 1 \mid \bigwedge_{j=1}^{i-1} \Ind_j = 1] \leq (1+\varepsilon) s/n$.
    By the chain rule of conditional probability, we have 
    $\Pr[\bigwedge_{j=1}^{i} \Ind_j = 1] = \Pr[\Ind_i = 1 \mid \bigwedge_{j=1}^{i-1} \Ind_j = 1] \cdot \Pr[\bigwedge_{j=1}^{i-1} \Ind_j = 1] = \Pr[\Ind_1 = 1] \cdot \prod_{j=2}^{cs}\Pr[\Ind_j = 1 \mid \bigwedge_{k=1}^{j-1} \Ind_k = 1] \leq ((1+\varepsilon) s/n)^i$. Note that in an analogous coin-flipping experiment with $c s$ coins, the probability that you get at least $(1-\Phi) c s$ coins is upper bounded by the probability that you get $(1-\Phi) c s$ heads and leave other coins unobserved. Thus, $\Pr[\mathcal{L}(S) \geq (1-\Phi) c s] \leq \binom{cs}{(1-\Phi) c s} ((1+\varepsilon) s/n)^{(1-\Phi) c s}$.

    From the above inequality, we get $p^* = \Pr[\exists S, |S| = s \; \texttt{and} \; \mathcal{L}(S) \geq (1-\Phi) c s] \leq \binom{n}{s} \binom{cs}{(1-\Phi) c s} ((1+\varepsilon)s/n)^{(1-\Phi) c s}$. By Lemma \ref{lem:probaInequality}, 
    $p^* \leq 1/n^2$ for large enough $n$ and some suitably chosen integer $c \geq 1$ and constant $\varepsilon \geq 0$ (where $\varepsilon$ can be made as small as required by taking $\ell = O(\log n)$ large enough). It suffices to union bound over all $n/2$ possible sizes for $S$ (for $s \in \{1,\ldots,n/2\}$) to get that the resulting graph has constant conductance $\Phi$ with probability at least $1 - 1/n$.  
\end{proof}

\subsection{Computing Graph Sketches via Gossip}
\label{subsec:push-sum}

Consider a node subset $S \subset V$ and a set of (peer-to-peer) edges $E' \subseteq V^2$, such that the graph $G'= (S,E')$ is an $O(1)$-bounded degree expander graph (and thus connected), whose maximum degree $d$ is known to all nodes (in $S$). We describe the \AggregateSketchVector{} primitive run on $G'$. We assume that all nodes in $S$ know the minimum ID among the nodes in $S$, that all nodes in $S$ have computed a common graph sketch matrix $M_G$, with certain properties (see Lemma~\ref{lem:aggregating-sketch-vectors-works}), and 
that each node $u \in S$ has computed the corresponding sketch vector $\sketchVect(u)$. Primitive \AggregateSketchVector{} computes 
the aggregate sketch vector $\sketchVect(S) = \sum_{u \in S} \sketchVect(u)$. 

\medskip

\noindent \textbf{Description.}
Each node $u \in S$ creates two sketch vectors, a positive sketch vector and a negative sketch vector, denoted by $\sketchVect^+(u)$ and $\sketchVect^-(u)$, respectively. $\sketchVect^+(u)$ holds all positive value entries of $\sketchVect(u)$ and the remaining entries are zero. $\sketchVect^-(u)$ is similarly defined. Define $\sketchVect^+(S) = \sum_{u \in S} \sketchVect^+(u)$ and $\sketchVect^-(S) = \sum_{u \in S} \sketchVect^-(u)$. All nodes $u$ run two instances of \PushSum{} (see~\cite{KDG03}) for $T_s = O(\log n \log (nx) )$ phases, where each phase consists of $T_p = O(\log^2 n)$ rounds and $x$ comes from the graph sketching matrix $M_G$ (see Lemma~\ref{lem:aggregating-sketch-vectors-works}), to obtain $\sketchVect^+(S)$ and $\sketchVect^-(S)$. Note that \PushSum{}, as described in~\cite{KDG03}, was described for a  complete graph and is a gossip-style technique to compute aggregate functions (like average, sum, etc.) in a network. The description below shows how to simulate it on an $O(1)$-bounded degree expander to compute the sum. 

We describe an instance of \PushSum{} from the perspective of node $u$ to obtain $\sketchVect^+(S)$. The process is similar to obtain $\sketchVect^-(S)$. At the end of each phase $t$, each node $u$ maintains some weight value $w_{t,u}^+$ and some estimate of the average of the sketches $s_{t,u}^+$. After a sufficiently long time $t^* = T_s T_p + 1$ has passed, the ratio $s^+_{t^*,u}/w^+_{t^*,u}$ will be an approximation of  $\sketchVect^+(S)$. 

Initially, the minimum ID node $\min$ sets its weight $w_{0,\min}^+ = 1$ and the remaining nodes $u$ set their weights $w_{0,u}^+ = 0$. Initially, every node $u$ sets $s_{0,u}^+ = \sketchVect^+(S)$. We assume that at the end of the phase $0$ (i.e., before the algorithm begins), each node $u$ sends the pair $(w^+_{0,u}, s^+_{0,u})$ to itself. 

In each phase $t \geq 1$ of \PushSum{}, each node $u$ does the following. Let $\lbrace (\hat{s_r}, \hat{w_r}) \rbrace$ be all pairs sent to node $u$ at the end of phase $t-1$. Node $u$ computes $s^+_{t,u} = \sum_r \hat{s_r}$ and $w^+_{t,u} = \sum_r \hat{w_r}$.\footnote{Recall that $s^+_{t,u}$ is some vector. The sum of vectors denotes the sum of the elements for each index of the vectors.} 
Node $u$ constructs the pair $( \frac{1}{2}s^+_{t,u},\frac{1}{2}w^+_{t,u} )$ and sends it to itself and a node $v \in S$ chosen uniformly at random as follows. 
Node $u$ initiates a lazy random walk of length $O(\log n)$ 
within $G'$ carrying the message $( \frac{1}{2}s^+_{t,u},\frac{1}{2}w^+_{t,u} )$. During these $T_p$ rounds, node $u$ helps other messages continue their random walks. 
After $T_s$ phases are over, i.e, at time $t^* = T_s T_p + 1$ rounds, an approximation of $\sketchVect^+(S)$ is constructed at node $u$ as $s^+_{t^*,u}/w^+_{t^*,u}$. 
The correct values of $\sketchVect^+(S)$ can be recovered from $s^+_{t^*,u}/w^+_{t^*,u}$ by rounding each element in $s^+_{t^*,u}/w^+_{t^*,u}$ to the nearest legal value, i.e., the value that an element in the sketch vector can take.

Now, each node has computed 
$\sketchVect^+(S)$ and $\sketchVect^-(S)$, and computes the output 
$\sketchVect(S)$  as $\sketchVect(S) = \sketchVect^+(S) + \sketchVect^-(S)$. 

\noindent \textbf{Analysis.}
In order to capture the properties of \AggregateSketchVector{}, we must first look at the properties of \PushSum{}. The following lemma from~\cite{KDG03} captures the properties of \PushSum{}. To bridge the notation, notice that each node $u$ contains some initial sketch vector $x_u = \sketchVect^+(u)$ (and for another instance of \PushSum{}  $x_u = \sketchVect^-(u)$) and our goal is to calculate $\sum_{j \in S} x_j$.

\begin{lemma}[Theorem~3.1 in~\cite{KDG03}]\label{lem:push-sum}
 \begin{enumerate}
     \item With probability at least $1 - \delta$, there is a time $t_0 = O(\log n + \log\frac{1}{\epsilon} + \log\frac{1}{\delta})$, such that for all times $t \geq t_0$ and all nodes $u$, the relative error in the  estimate of the average at node $u$ is at most $\epsilon \cdot \frac{\sum_j |x_j|}{|\sum_j x_j|}$ (where the relative error is $\frac{1}{|\sum_j x_j|} \cdot |\frac{s_{t,u}}{w_{t,u}} - \frac{1}{n} \cdot \sum_j x_j|$). In particular, the relative error is at most $\epsilon$ whenever all values $x_j$ have the same sign.
     \item The sizes of all messages sent at time $t$ are bounded by $O(t+\max_j \text{bits}(x_j))$ bits, where $\text{bits}(x_j)$ denotes the number of bits in the binary representation of $x_j$.
 \end{enumerate}
\end{lemma}

As mentioned in~\cite{KDG03}, we can get the sum of the values of all nodes instead of just the average, by setting the initial weights such that only one node has weight $1$ and the remaining have weight $0$, as we do in \AggregateSketchVector{}.

It should be noted that the original \PushSum{} was designed for a complete graph, where each node $u$ could sample one of the nodes of the graph uniformly at random. We simulate this process on an $O(1)$-bounded degree expander by having node $u$ choose another node as the result of a lazy random walk that is run for the mixing time. This guarantees us that we may sample a node in $S$ nearly uniformly at random.

Since this requires multiple nodes initiating and facilitating random walks simultaneously, we make use of the following lemma, adapted from a lemma in~\cite{jacm13}. We note that it is for the traditional CONGEST model, where every node can communicate with all of its neighbors in a given round, unlike the current setting. We explain in the analysis of the final lemma of this section, how the following lemma can easily be adapted to the current model.

\begin{lemma}[Adapted from Lemma 3.2 in~\cite{jacm13}]\label{lem:random-walk-time}
In the traditional CONGEST model, let $G=(V,E)$ be an undirected graph and let each node $v\in V$, with degree $d(v)$, initiate  $\eta d(v)$ random walks, each
of length $\lambda$. Then all walks reach their destinations in $O(\eta \lambda \log n)$ rounds with high probability.
\end{lemma}

The following lemma shows that \AggregateSketchVector{} works as desired to help us reconstruct the aggregate of the initial sketch vectors in the desired time.

\begin{lemma}\label{lem:aggregating-sketch-vectors-works}
Assume that each node $u$ has a graph sketch vector $\sketchVect(u)$ computed using a graph sketching matrix $M_G$ that satisfies 
the following properties:
\begin{itemize}
\item For all nodes $u$, elements in $\sketchVect(u)$ belong to the same range of values $[L, U]$, where $L,U \in 	\mathbb{R}$, and $U-L = x \neq 0$.
\item The range of values taken by elements in the sketch vector, $Range$, is some totally ordered countable set of numbers such that the minimum distance between any two numbers is at least some constant $c > 0$, i.e., $\forall u, v \in Range$ if $u \neq v$, then $|u - v| \geq c$.
\end{itemize}
If all nodes $u\in S$ participate in \AggregateSketchVector{} for $O(\log^3 n \log (nx))$ rounds, then each node outputs the aggregated sketch vector $\sketchVect(S) = \sum_{u \in S} \sketchVect(u)$ with high probability.
\end{lemma}

\begin{proof}
We first argue that \PushSum{} is faithfully simulated. Each round of the original \PushSum{} corresponds to $O(\log n)$ phases of the process as described here. We first describe the end result of one phase. In one phase, we ensure that a lazy random walk starting at some node $u$ is run for enough rounds to reach mixing time. In an $O(1)$-bounded degree expander, the mixing time of one single lazy random walk is $O(\log n)$ rounds. However, since each node simultaneously initiates a single lazy random walk of length $O(\log n)$, by Lemma~\ref{lem:random-walk-time}, we see that we need $O(\log^2 n)$ rounds to ensure that they all complete. Furthermore, it should be noted that Lemma~\ref{lem:random-walk-time} applies to the traditional CONGEST model. However, in a bounded degree graph with maximum degree $d$, one round of the traditional CONGEST model can be simulated in $d$ rounds in the current model. Since $d = O(1)$, we incur no overhead with respect to the lemma and see that the $O(1)$ random walks initiated from each node, each run for $O(\log n)$ time, require in total $O(\log^2 n)$ rounds to complete. Thus in one phase, some node from $G'$ will be sampled (nearly) uniformly at random.

To faithfully simulate \PushSum{}, we want the sampled node to not be the starting node. However, in order to ensure that some node, other than the starting node, is sampled uniformly at random, we need to run $O(\log n)$ phases of this protocol. By a simple Chernoff bound, we can see that with high probability, even when $|S| = O(1)$, a random walk starting at some node $u$ will sample some node that is not $u$.

Notice that by separating each node's sketch vector into two vectors corresponding to positive values and negative values and using \PushSum{} to find the aggregate of each set of vectors separately, we ensure that the signs of values being aggregated is the same. From Lemma~\ref{lem:push-sum}, we see that the relative error of any entry is thus at most $\epsilon$. By setting $\epsilon = 1/nx$, we see that the output of each element of the sketch vector is within $1/n$ of the actual value. Since each element of the sketch vector can only take values such that the distance between values is some constant $c$, it is easy to see that rounding the element to the nearest legal value gives the correct value (since the outputted value will only be at most some $1/n \ll c/2$ from a legal value).

Finally, in order to ensure that these guarantees are with high probability, it is sufficient to set $\delta = 1/n$ in Lemma~\ref{lem:push-sum}. Substituting the values of $\epsilon$ and $\delta$ in Lemma~\ref{lem:push-sum}, we get the corresponding run time.
\end{proof}

In our main algorithm in Section~\ref{sec:overlayConstruction}, \AggregateSketchVector{} is used to compute the aggregate of the sketch vectors that are obtained via the graph sketching matrix guaranteed by Lemma \ref{lem:graph-sketch-works} in Section~\ref{subsec:graph-sketch}. 
That matrix satisfies the requirements of Lemma~\ref{lem:aggregating-sketch-vectors-works} and $x = O(n^r)$, for some positive constant $r \geq 1$, and $c=1$. Thus, we have the following corollary.

\begin{corollary}\label{cor:aggregate-sketch-vector-works}
Consider any node subset $S \subset V$ and set of (possibly peer-to-peer) edges $E' \subseteq V^2$, such that the graph $G'= (S,E')$ is an $O(1)$-bounded degree expander graph. 
Assume that each node $u$ has a graph sketch vector $\sketchVect(u)$ computed using the graph sketching matrix guaranteed by Lemma \ref{lem:graph-sketch-works} in Section~\ref{subsec:graph-sketch}. 
If all nodes $u\in S$ participate in \AggregateSketchVector{} for $O(\log^4 n)$ rounds, then each node will obtain the aggregated sketch vector $\sketchVect(S) = \sum_{u \in S} \sketchVect(u)$ with high probability.
\end{corollary}

%% file: overlayConstruction.tex
Let $G = (V,E)$ be the original graph and let $\Phi \in (0,1/10]$ be the desired conductance. We show how to transform any arbitrary graph $G = (V,E)$  (even with large degree) into an expander overlay network with conductance $\geq \Phi$ and bounded degree. We assume the harder case of $\Phi = \Omega(1)$, or in other words, of building an overlay network with conductance $\Omega(1)$. 

Initially, each node forms its own (high conductance) cluster, and the set of intra-cluster edges $E_0$ is empty. In each stage $i \in [1,\st]$, we compute a new set of edges $E_i \subseteq V^2$ 
that merge the stage's starting clusters into fewer and larger-sized (high conductance) clusters.
In fact, the merging reduces the number of clusters by half with constant probability.
After $\st = O(\log n)$ stages, the resulting edge set $E_{\st}$ generates an expander graph $G(E_{\st}) = (V,E_{\st})$ with high conductance (i.e., at least $\Phi$).

\subsection{Algorithm Description}

The algorithm runs in $\st = O(\log n)$ stages, each consisting of three steps. We ensure the following invariant holds at the start of each stage $i \in [1,\st]$: the graph generated by $E_{i-1}$, $G(E_{i-1}) = (V,E_{i-1})$, is decomposed into clusters (i.e., connected components) with constant conductance $\Phi$ and constant maximum degree. (Note that the initial edge set $E_0$ is empty and thus the invariant holds trivially.) 
\medskip

\noindent \textbf{First Step.}
Nodes spread, within each cluster of $G(E_{i-1})$, the minimum ID and an associated $O(\log^2 n)$ random bit string, by executing a PUSH style information spreading algorithm (e.g., see~\cite{FPRU90} or \cite{Giakkoupis2011}) over (each cluster of) $G(E_{i-1})$ for $T_g = O( \Phi^{-1}\log n )$ rounds. More concretely, nodes do the following. Initially, each node picks an $O(\log^2 n)$ bit random string. Then, in each round, each node chooses one random neighbor in $G(E_{i-1})$ and sends a message containing $\langle$minimum ID seen so far, associated random string$\rangle$ to it. Since $G(E_{i-1})$ decomposes into clusters of maximum degree $O(1)$ and of conductance $\geq \Phi$, and thus of diameter $O(\Phi^{-1} \log n)$, the result of~\cite{FPRU90} implies that $T_g$ rounds is sufficient to spread, within each cluster, that cluster's minimum ID and its associated $O(\log^2 n)$ bit random string.

\medskip

\noindent \textbf{Second Step.}
Now, all nodes within some cluster $V_j$ of $G(E_{i-1})$ know the minimum ID of that cluster and the associated $O(\log^2 n)$ bit random string. Then, nodes use this shared randomness to sample one edge per cluster, using graph sketches. We describe how this is done in the next paragraph. Each such sampled edge is an inter-cluster edge (i.e., its endpoints are in different clusters) with constant probability. These sampled inter-cluster edges --- the set of which is denoted by $E_i^c$ --- allow to merge clusters of $G(E_{i-1})$, reducing them by a constant fraction with constant probability. However, although each cluster samples a single edge, each node may have $\omega(1)$ incident edges in $E_i^c$: for example, if many clusters sample edges incident to one particular node. Thus, we finish the step by computing a second edge set $E_i^b$ such that $G(E_i^b)$ has bounded degree and preserves the connectivity of $G(E_i^c)$. 

Let $v$ be any node within some cluster $V_j$ of $G(E_{i-1})$. First, each node $v$ computes its sketch vector $\sketchVect(v)$ (using the cluster's shared random string to generate a graph sketching matrix, see Subsection \ref{subsec:graph-sketch}) and runs Procedure \AggregateSketchVector{}. Its output is the aggregated sketch vector $\sigma(v) = \sum_{v \in V_j} \sketchVect(v)$. Next, $v$ samples an edge using $\sigma(v)$ (see Subsection \ref{subsec:graph-sketch}), and this edge is an inter-cluster edge with some constant probability $\delta$. If $v$ is an endpoint of that edge, then it contacts the other endpoint node, they exchange information on their cluster's minimum ID and $v$ drops the edge if these IDs are identical. In other words, all sampled edges whose endpoints come from different clusters (i.e., exchanged different IDs) are added to the edge set $E_i^c$, whereas others are simply ``dropped''.

Now, it remains to compute $E_i^b$, by applying a simple degree reduction procedure on $E_i^c$. Note that for each edge in $E_i^c$, one node belongs to the cluster that sampled that edge. That node is said to own the edge and one can, for the sake of the procedure's description, orient the edge from the owner node to the other endpoint.
Then, nodes do the following in two rounds. Initially, nodes add all incident edges in $E_i^c$ to $E_i^b$. Any node $u$ with at least 3 incoming incident edges in $E_i^b$ --- the set of their owners is denoted by $N_u$ --- locally computes an arbitrary undirected cycle over $N_u \cup \{u\}$, and locally replaces these incident edges in $E_i^b$ with its two incident edges in that cycle. Then, each node that owns an edge contacts the other endpoint, and they exchange the result of their local computations (if any). Finally, nodes that received a cycle locally replace their owned edge in $E_i^b$ with their two incident edges in that cycle.  

\medskip

\noindent \textbf{Third Step.}
Finally, we transform each cluster of $G(E_{i-1} \cup E_i^b)$ into an $O(1)$-bounded degree expander graph with constant conductance $\Phi$. More concretely, we compute a set of edges $E_i$ for which: (i) each cluster of $G(E_i)$ is a cluster of $G(E_{i-1} \cup E_i^b)$, and (ii) each cluster of $G(E_i)$ is an $O(1)$-bounded degree expander graph with constant conductance $\Phi$.

To do so, we first run Procedure \CreateExpander{} (described in Subsection \ref{subsec:createExpander}) for $O(\log^2 n)$ rounds. This computes, for each cluster of $G(E_{i-1} \cup E_i^b)$, an $O(\log n)$-regular expander graph with constant conductance $\Phi$. Next, we run Procedure \ExpanderDegreeReduction{} (described in Subsection \ref{subsec:degreeReductionExpander}) for $O(\log^3 n)$ rounds to reduce the degree of these expander graphs to some constant. More precisely, this procedure computes, for each $O(\log n)$-regular expander graph, an $O(1)$-bounded degree expander graph that also has constant conductance $\Phi$.
Adding together the edges of all these $O(1)$-bounded degree expander graphs gives the edge set $E_i$.

\subsection{Analysis}
\label{sec:analysis}

To prove the correctness of our overlay construction algorithm and bound its round and message complexities (see Theorem \ref{thm:mainTheorem}), we first give a series of lemmas that rely upon the following invariant: for any stage $i \in [1,\st]$, all clusters of $G(E_{i-1})$ are $O(1)$-bounded degree expander graphs with constant conductance $\Phi$. We start by showing that each edge sampled using graph sketching is an inter-cluster edge with constant probability.

\begin{lemma}
\label{lem:interClusterEdgeSampling}
For any stage $i \in [1,\st]$, assume all clusters of $G(E_{i-1})$ are $O(1)$-bounded degree expander graphs with constant conductance $\Phi$. Then, each cluster of $G(E_{i-1})$ samples an inter-cluster edge with probability at least $3/4$.
\end{lemma}

\begin{proof}
Since we use the information spreading algorithm of \cite{Giakkoupis2011} (see Theorem 1 in \cite{Giakkoupis2011}), and all clusters of $G(E_{i-1})$ are $O(1)$-bounded degree expander graphs with constant conductance $\Phi$, then for each such cluster, the minimum ID and the associated $O(\log^2 n)$ bits random string is spread to all of that cluster's nodes in $O(\Phi^{-1} \log n)$ rounds w.h.p.

Next, consider any cluster $V_j$ of $G(E_{i-1})$. Recall that each node $v \in V_j$ computes its sketch vector $\sketchVect(v)$ initially and uses it as input for Procedure \AggregateSketchVector{}. Its output is $\sigma(v) = \sum_{u \in V_j} \sketchVect(u)$, the aggregate of sketch vector within cluster $V_j$, with high probability by
Corollary~\ref{cor:aggregate-sketch-vector-works}.
Then, by choosing $\delta = 1/4$, each node can sample an inter-cluster edge from this aggregate vector with constant probability $1 - \delta = 3/4$, by Lemma~\ref{lem:graph-sketch-works}. 
\end{proof}

As a result of the above lemma, many sampled edges are inter-cluster with constant probability. Thus, the addition of these sampled edges significantly reduces the number of clusters --- in fact, by a constant fraction --- with constant probability.

\begin{lemma}
\label{lem:secondStepMerge}
For any stage $i \in [1,\st]$, assume all $c \geq 1$ clusters of $G(E_{i-1})$ are $O(1)$-bounded degree expander graphs with constant conductance $\Phi$. If $G(E_{i-1})$ has $c > 1$ clusters, then $G(E_{i-1} \cup E_i^c)$ has at most $3c/4$ clusters with probability at least 1/2. 
\end{lemma}

\begin{proof}
     All clusters of $G(E_{i-1})$ are $O(1)$-bounded degree expander graphs with constant conductance $\Phi$ (from the lemma's assumption). Then, by Lemma \ref{lem:interClusterEdgeSampling}, each cluster of $G(E_{i-1})$ samples an inter-cluster edge with probability $3/4$. By linearity of expectation, in expectation $3c/4$ of the sampled edges are inter-cluster edges, or equivalently, in expectation $c/4$ sampled edges are intra-cluster edges. Applying Markov's inequality, the probability that more than $c/2$ sampled edges are intra-cluster is at most 1/2, or equivalently, the probability that at least $c/2$ sampled edges are inter-cluster is at least 1/2. Any of these inter-cluster edges (say, from $V_j$ to $V_{j'}$) allows to merge two clusters and reduce the number of clusters by 1, unless an inter-cluster edge from $V_{j'}$ to $V_{j}$ was also sampled. Hence, if there are at least $a$ inter-cluster edges, then $G(E_{i-1} \cup E_i^c)$ has at most $c - a/2$ clusters. Thus, $G(E_{i-1} \cup E_i^c)$ has at most $3c/4$ clusters with probability at least 1/2.
\end{proof}

Note that the sampled inter-cluster edges may generate a graph with large degree. Next, we prove that the degree reduction procedure used in the second step is correct.

\begin{lemma}
\label{lem:secondStepSparsify}
For any stage $i \in [1,\st]$, $G(E_i^b)$ has maximum degree at most 4 and preserves the connectivity of $G(E_i^c)$, that is, for any edge $(u,v) \in E_i^c$, $u$ and $v$ are connected in $G(E_i^b)$.
\end{lemma}

\begin{proof}
Consider an arbitrary node $u$. As in the algorithm description, for each edge of $E_i^c$, assume for the sake of the proof that it is directed from the sampling cluster outwards. (Note that $E_i^c$ contains no edges with both endpoints in the same cluster.) We show that $u$ is incident to at most 2 edges in $E_i^b$ due to incoming edges in $E_i^c$, and to at most 2 other edges in $E_i^b$ due to outgoing edges in $E_i^c$; thus, $u$ has a degree of at most 4 in $G(E_i^b)$. On the one hand, if $u$ has at least 3 incoming edges in $E_i^c$, then these edges are replaced by only 2 edges (from the cycle built by $u$ in the second half of the second step) in $E_i^b$. On the other hand, each cluster, and thus node, is incident to at most one outgoing edge in $E_i^c$ (by the algorithm definition). This outgoing edge may be replaced by at most two edges in $E_i^b$ (when the outgoing edge's other endpoint has more than 3 incoming edges in $E_i^c$). 

Finally, it is straightforward to show that $G(E_i^b)$ preserves the connectivity of $G(E_i^c)$. Indeed, any edge $(u,v) \in E_i^c$ that does not remain in $E_i^b$ is replaced by a cycle in $E_i^b$ (locally computed by either $u$ or $v$) that includes both $u$ and $v$. 
\end{proof}

Next, we show that the invariant is maintained, and that the stage reduces the number of clusters by a constant fraction with constant probability.

\begin{lemma}
\label{lem:clusterDecrease}
For any stage $i \in [1,\st]$, assume all $c \geq 1$ clusters of $G(E_{i-1})$ are $O(1)$-bounded degree expander graphs with constant conductance $\Phi$. Then, all $c' \leq c$ clusters of $G(E_{i})$ are $O(1)$-bounded degree expander graphs with constant conductance $\Phi$. Moreover, if $G(E_{i-1})$ has $c > 1$ clusters, then $G(E_{i})$ has at most $3c/4$ clusters with probability 1/2. 
\end{lemma}

\begin{proof}
    Consider some stage $i \in [1,\st]$. Let the $c \geq 1$ clusters of $G(E_{i-1})$ be denoted by $V_1,\ldots,V_c$. We first provide some properties about $G(E_{i-1} \cup E_i^b)$. To start with, $G(E_{i-1})$ has constant maximum degree (from the lemma's assumption) and thus by Lemma \ref{lem:secondStepSparsify}, $G(E_{i-1} \cup E_i^b)$ also has constant maximum degree, denoted by $d$. 
    Second, each cluster $V_i$ of $G(E_{i-1})$ is part of (i.e., a subset of) some cluster in $G(E_{i-1} \cup E_i^b)$, as the latter graph only has additional edges. Third, if $G(E_{i-1})$ has $c > 1$ clusters, then $G(E_{i-1} \cup E_i^b)$ has at most $c/2$ clusters with constant probability. Indeed, $G(E_{i-1} \cup E_i^c)$ has at most $3c/4$ clusters with probability at least 1/2. Since $G(E_{i-1} \cup E_i^b)$ preserves the connectivity of $G(E_{i-1} \cup E_i^c)$ by Lemma \ref{lem:secondStepSparsify}, $G(E_{i-1} \cup E_i^b)$ also has at most $3c/4$ clusters with probability at least 1/2.

    Finally, recall that the edge set $E_i$ is obtained by nodes running Procedure \CreateExpander{} (with parameters $\Phi$ and $d$) followed by Procedure \ExpanderDegreeReduction{} on each cluster of $G(E_{i-1} \cup E_i^b)$. First, note that Procedures \CreateExpander{} and \ExpanderDegreeReduction{} do not disconnect any clusters of $G(E_{i-1} \cup E_i^b)$, and thus in particular $G(E_{i})$ has at most as many clusters as $G(E_{i-1} \cup E_i^b)$. Moreover, by Lemma \ref{lem:GotteResult} and Theorem \ref{thm:expanderDegreeReductionGuarantees}, each cluster of $G(E_{i})$ is an $O(1)$-bounded degree expander graph with constant conductance $\Phi$. This completes the proof.   
\end{proof}

Note that by definition, any graph has at least one cluster. The above lemma implies that within $O(\log n)$ stages, we obtain a graph with exactly one high-conductance cluster, and thus solve the overlay construction problem --- see the following theorem.

\begin{theorem}
\label{thm:mainTheorem}
The overlay construction problem can be solved with high probability in $O(\log^5 n)$ rounds and $\tilde{O}(n)$ messages.
\end{theorem}

\begin{proof}
To start with, for any given stage $i \in [1,\st]$, the stage is said to be \emph{successful} if $G(E_{i})$ either (i) has a single cluster or (ii) has less than 3/4 as many clusters as $G(E_{i-1})$. By Lemma \ref{lem:clusterDecrease} (and a simple induction on $i$), each stage is successful with probability at least 1/2. Hence, for a large enough number $\st = O(\log n)$ of stages, a simple application of Chernoff bounds imply that there are at least $\log_{4/3} n$ successful stages. Thus, $G(E_{\st})$ has a single cluster. Moreover, by Lemma \ref{lem:clusterDecrease} again, $G(E_{\st})$ is an $O(1)$-bounded degree expander graph with constant conductance $\Phi$. The correctness follows.

The round complexity of $O(\log^5 n)$ is straightforward: $O(\log n)$ stages each take $O(\log^4 n)$ rounds. 
The message complexity follows directly from the time complexity, the fact that communication is always done on graphs of degree at most $O(\log n)$, and that messages of size $O(\polylog n)$ suffice throughout the algorithm (both to share the random string in the first step and  by Lemma \ref{lem:push-sum}, during Procedure \PushSum{} in the second step).
\end{proof}

%% file: experiment.tex
The proposed overlay construction protocol is implemented in a sequential simulation to study properties of the algorithm for a few different types of low-conductance graphs. We  study the number of rounds and the conductance of the resulting graphs. The simulation follows the algorithm's steps with some small deviations.

The algorithm is implemented in sequential simulated form using Python and the graph library graph\_tool \cite{peixoto_graph-tool_2014}, to study properties of the algorithm for a few low-conductance graphs with different properties, and study the number of rounds and estimate the conductance provided by the algorithm. The types of graphs tested by the simulator include a high-diameter cycle graph \textit{circle-10000}, a graph on a square grid \textit{grid-50-50}, a randomly-generated preferential attachment graph \textit{barabasi-2000-2-2}, as well as modestly sized real-world graphs with differing topologies: graphs modeling disease contagion \textit{kissler}, social network attachment \textit{twitch}, and citation networks \textit{wiki}. Table \ref{fig:exp-results} summarizes the results; full details on the implementation, graphs studied, and observations of performance are included in Section \ref{app:appendixExperimental} of the Appendix.

\begin{table}[h]
\small
    \begin{tabular}{ccccccccc}
        \toprule
        Graph & $n$ & Phases && $D_G$ & $D_{G_E}$ && $\Phi_{G}$ & $\Phi_{G_E}$ \\
        \cmidrule{1-3}\cmidrule{5-6}\cmidrule{8-9}
        \textit{circle-10000} & 10000 & 6 && 1111 & 6 && 0.068 & 0.453\\
        \textit{grid-50-50} & 2500 & 5 && 98 & 4 && 0.148 & 0.449 \\
        \textit{barabasi-2000-2-2} & 2000 & 3 && 5 & 4 && 0.4 & 0.451 \\
        \textit{wiki} & 2277 & 3& & 16 & 4 && 0.08 & 0.450 \\
        \textit{twitch} & 7126 & 3 && 10 & 5 && 0.143 & 0.452 \\
        \textit{kissler} & 409 & 3 && 9 & 3 && 0.2 & 0.446 \\

        \bottomrule
    \end{tabular}
    \caption{Table showing simulation results  on  various input graphs (denoted by $G$) and
    the corresponding graphs output by the protocol (denoted by $G_E$). $n$ is the number of nodes of $G$. Phases
    denotes the number of phases of the overlay construction protocol   that were required to produce $G_E$. $D_G$ and $D_{G_E}$ are lower-bound estimates of the graph diameter of $G$ and $G_E$. $\Phi_{G}$ and $\Phi_{G_E}$ are upper-bound estimates of conductance of $G$ and $G_E$. }
    \label{fig:exp-results}
\end{table}

The conductance of the input and the final overlay $G$ and $G_E$ are each estimated using $O(n)$ sampled graph cuts, to provide an upper-bound estimate on the actual graph conductances $\Phi_G$ and $\Phi_{G_E}$.
The table shows for each input graph these conductance estimates as well as the number of phases
of the protocol required for the algorithm to terminate. The table also gives the estimate of the diameters $D_G$ and $D_{G_E}$ of the initial and final graphs respectively, given by the pseudo-diameter as calculated for a sampling of nodes. The results show that the conductance of the constructed
graph is likely significantly higher compared to the starting graph and is close to $0.5$ which is essentially
the best possible value for a constant-degree random graph.
The results also show that the diameter of the final expander is roughly in line with expectations of an $O(\log{n})$ bound, and that the number of rounds, conductance, and diameter of $G_E$ are independent of the edge density of the initial graph.

%% file: conclusion.tex
In this paper, we presented the first distributed  overlay construction protocol
that is fast (taking $O(\log^5 n)$ rounds) as well as taking significantly less communication (using
$\tilde{O}(n)$ messages, regardless of the number of edges of the initial graph). The protocol
assumes the P2P-GOSSIP model which uses gossip-based communication (which is very lightweight)
and the reconfigurable nature of P2P networks.
Both bounds are essentially the best possible. Our result also implies that the distributed complexity
of solving fundamental problems such as broadcast, leader election, and MST construction is significantly
smaller in the P2P-GOSSIP model compared to the standard CONGEST model.

Several open questions remain. One is to improve the round complexity of our protocol.  In particular,
can we improve the round complexity to  $O(\log^2 n)$ rounds while keeping $\tilde{O}(n)$ communication? Another interesting follow up work is to adapt our protocol to work under a churn setting. A third interesting research direction is to investigate the complexity of other fundamental problems such as computing shortest paths in the P2P-GOSSIP model.

%% file: experiment-appendix.tex
The algorithm is implemented in sequential simulated form using Python and the graph library graph\_tool \cite{peixoto_graph-tool_2014}, to study properties of the algorithm for a few low-conductance graphs with different properties, and study the number of rounds and estimate the conductance provided by the algorithm.

\subsection{Simulation Implementation Details}

The simulator begins by dividing the input graph $G$ into $n$ distinct components of size 1, and then progresses through a series of $O(\log{n})$ rounds, each divided into sampling, merging, sparsification, and expansion stages.

Each distinct component generates an $\ell_0$ sketch that samples a single inter-cluster edge with high probability. These sketches are implemented using $k$-wise independent hash functions with sparse recovery using the techniques described in \cite{cormode}.
The $\ell_0$ sketch for each cluster is computed directly by updating the sketch with all incident edges of each node in the cluster, rather than through \PushSum{}. On the occasion that the $\ell_0$-sampling fails for a particular component, it is retried. This does not occur with high probability.

After the components are merged, the set of sampled edges is sparsified to reduce the degree of any nodes $u$ incident to a set of four or more sampled edges (denoted as $N_u$ in the algorithm description), by adding edges to create an undirected cycle over $N_u \cup \{u\}$, and removing all but two of the sampled edges incident to $u$.

Finally, the procedures \CreateExpander{} and \ExpanderDegreeReduction{} are performed to create overlays on each cluster such that the result a $O(1)$-bounded degree expander. The implementation generates the values $\Lambda = \log{n}$ and $\Delta = 2d\Lambda$ for the \CreateExpander{} step on a cluster of size $n$, with a constant number of $5$ overlay iterations and constant length random walks with length $\ell = 13$. For \ExpanderDegreeReduction{}, each node produces $10$ tokens and accepts up to $40$, yielding an expander graph with maximum degree at most $50$, and the length of the random walks was set at $\ell = 2\log{n}$.

For smaller clusters with very low degree nodes, it is sometime the case that the random walks on benign graphs will disconnect the graph: this should occur with low probability, but it is observed in practice on the graphs tested. The simulation handles this by checking that each cluster is connected, and if it is not, the nodes disconnected from the largest connected component restore the edges that they had in the previous overlay, or in the initial benign graph if disconnected in the first overlay.

The new merged components are each individually $O(1)$-degree expander graphs, and become subject to the next round of the algorithm.

\subsection{Graphs Under Consideration in Detail}

\textit{circle-10000} is a "fat" cycle graph with $n = 10000$ where each node is connected to the nine following nodes, and the nine terminal nodes are then attached similarly to the initial nodes. Thus, it is a circular 18-regular graph and is characterized by its high diameter.

\textit{grid-50-50} is a graph on a square grid with each side 50 nodes wide. It has 2500 nodes, and the maximum degree is 4. This graph thus has large diameter and low conductance as well.

\textit{barabasi-2000-2-2} is a random graph generated through a linear Barabási-Albert model with $m=2$ \cite{DBLP:journals/corr/abs-2110-00287}. Starting from a pair of nodes, 1998 nodes are added,
each bringing two additional edges. These edges connect to previous nodes using preferential attachment, where the probability of attaching to a previous node is proportional to the degree of the node: the probability that an incoming edge will connect to a node $v$ of degree $d_G(v)$ is given $p(v) = \frac{d_G(v)}{\sum_{i \in V}d_G(i)}$ This leads to a graph that has a minimum degree 2, and for which the degree follows a power law relationship $P(k) \sim k^{-3}$. This graph has very high degree subgraphs as well as nodes of very low degree.

The algorithm was also tested on the following graphs produced through real-world data:

\textit{kissler} is a graph representation of the proximity dataset of \cite{Kissler479154} and which has been processed according to the technique utilized in \cite{10.1371/journal.pone.0272739}. This graph describes close contact for individuals who utilized the BBC Pandemic Haslemere app, where the participants (nodes) have an edge corresponding to a point of contact within 5 meters for any length of time during the study period. The graph from this dataset has been restricted to its largest connected component for the purposes of this test.

\textit{twitch} is a graph representation of user relationships on the social network Twitch, where the nodes represent UK users of Twitch and edges represent mutual friendships between users, as collected in May 2018 by \cite{rozemberczki2019multiscale}. This graph has low diameter due to a small number of high-degree nodes, similar to a preferential attachment graph.

\textit{wiki} is a graph representation of a hyperlink network based around a specific topic, in particular, the Wikipedia article on chameleons, as collected in \cite{rozemberczki2019multiscale} from the English language Wikipedia in December 2018. This graph contains a small number of nodes with extremely high degree, and many subgraphs connected by single edges. Thus, it is poorly-connected and high diameter, with a correspondingly small conductance, yet also a large maximum degree.

\subsection{Observations}
\label{subsec:observations}

The analysis provided by the simulator gave subjective information about the performance of the algorithm under different conditions, as per the results in Table \ref{fig:exp-results}.

The particular choice of the length of the random walks of \CreateExpander{} was chosen by experimentation, to be sufficient to reduce the requirement of larger constants for $\Lambda$ and $\Delta$ that would increase the runtime of the algorithm greatly. 

The degree reduction of \ExpanderDegreeReduction{} was greatly beneficial to the runtime, as without it, the final expanders produced by \CreateExpander{} were observed to have significantly large degree, at minimum $\Delta/8 = O(\log{n})$, since by the final overlay stage, the graph had improved conductance sufficiently that most tokens are dispersed widely and are not discarded or returned to the original node. This was greater than the value of $d$ chosen after only two rounds. These $O(\log{n})$-max-degree intermediate expander graphs would be combined, and in applying \CreateExpander{} at this point, the resulting $\Delta$ for this graph would in practice be $2d_{max}\log{n} = O(\log^2{n})$. This was unnoticeable in the early rounds, but the final round which merged a small number of $O(n)$-size components, at this point, had quite sizable max degree. Since simulating the algorithm in a sequential setting required $O(n\Delta) = O(n \log^2{n})$ random walks for \CreateExpander{}, and sequentially computing the path of each random walk had complexity $O(k_e) = O(\log^2{n})$ where $k_e$ is the average degree of the cluster, not sparsifying leads to large polylog factors in the total runtime of \CreateExpander{} that adversely affect performance of the simulation.

The intermediate expander graphs, especially in early rounds, were close to being complete graphs. However, when the intermediate graphs were sparsified and thus became bounded-degree graphs, the resulting size of the graphs was more manageable as well. The process of \ExpanderDegreeReduction{} was efficient and effective with the given value of $\ell$, requiring significantly less time in the sequential model, as in practice all tokens became inactive in a single phase in virtually all attempts.

Given the same input constants, and roughly similar sizes of graphs (in terms of number of vertices), the resulting conductance upper-bounds on the expander graphs produced by the algorithm were extremely similar to each other, undoubtedly due to the similar bounds on the maximum degree. The conductance of the resulting expander was not dependent on the initial conductance; thus the lowest conductance graphs such as \textit{wiki-small} received a sizable improvement and the graphs already functionally an expander, like \textit{barabasi-2000-2-2}, showed rather little improvement. The total rounds of the algorithm varied most with graphs of smaller or higher diameter, with \textit{grid-50-50} requiring significantly more rounds than \textit{twitch} despite the latter graph having over twice the vertices.